\documentclass[a4paper,UKenglish,cleveref, autoref]{lipics-v2019}

\usepackage{wrapfig}
\usepackage{xspace}
\newcommand{\containment}{NP~$\subseteq$~coNP/poly\xspace}
\newcommand{\notcontainment}{NP~$\not \subseteq$~coNP/poly\xspace}
\newcommand{\vc}{\mathop{\textsc{vc}}}

\theoremstyle{plain}
\newtheorem*{proposition*}{Proposition}

\newtheorem*{theorem*}{Theorem}
\newtheorem*{lemma*}{Lemma}

\title{Preprocessing Vertex-Deletion Problems: Characterizing Graph Properties by Low-Rank Adjacencies} 

\titlerunning{Preprocessing Vertex-Deletion Problems}

\author{Bart M.P. Jansen}{Eindhoven University of Technology, The Netherlands}{b.m.p.jansen@tue.nl}{https://orcid.org/0000-0001-8204-1268}{}

\author{Jari J.H. de Kroon}{Eindhoven University of Technology, The Netherlands }{j.j.h.d.kroon@tue.nl}{https://orcid.org/0000-0003-3328-9712}{}

\authorrunning{B.\,M.\,P. Jansen and J.\,J.\,H. de Kroon}

\Copyright{Bart M.P. Jansen and Jari J.H. de Kroon}



\category{}

\relatedversion{}

\supplement{}

\funding{This project has received funding from the European Research Council (ERC) under the European Union's Horizon 2020 research and innovation programme (grant agreement No 803421, ReduceSearch).
\\
\includegraphics[width=2cm]{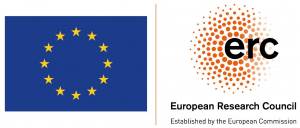}
} 

\acknowledgements{We would like to thank Fedor V. Fomin for hosting Jari in Bergen (Norway). } 

\nolinenumbers 

\hideLIPIcs  

\EventEditors{John Q. Open and Joan R. Access}
\EventNoEds{2}
\EventLongTitle{42nd Conference on Very Important Topics (CVIT 2016)}
\EventShortTitle{CVIT 2016}
\EventAcronym{CVIT}
\EventYear{2016}
\EventDate{December 24--27, 2016}
\EventLocation{Little Whinging, United Kingdom}
\EventLogo{}
\SeriesVolume{42}
\ArticleNo{23}

\usepackage{algorithm}
\usepackage{algorithmic}

\usepackage{todonotes}

\newcommand{\incf}[5]{\mathrm{inc}^{#1,(#4,#5)}_{(#2,#3)}}
\newcommand{\incp}[2]{\mathrm{inc}^{(#1,#2)}}
\newcommand{\incs}[3]{\mathrm{inc}^{#1}_{(#2,#3)}}
\newcommand{\incd}{\mathrm{inc}}

\newcommand{\defparproblem}[4]{
	\vspace{1mm}
	\noindent\fbox{
		\begin{minipage}{0.96\textwidth}
			\begin{tabular*}{\textwidth}{@{\extracolsep{\fill}}lr} \textsc{#1} & {\bf{Parameter:}} #3 \\ \end{tabular*}
			{\bf{Input:}} #2 \\
			{\bf{Question:}} #4
		\end{minipage}
	}
	\vspace{1mm}
}

\renewcommand{\vec}[1]{\mathbf{#1}}

\begin{document}

\maketitle

\begin{abstract}
We consider the \textsc{$\Pi$-free Deletion} problem parameterized by the size of a vertex cover, for a range of graph properties~$\Pi$. Given an input graph~$G$, this problem asks whether there is a subset of at most~$k$ vertices whose removal ensures the resulting graph does not contain a graph from $\Pi$ as induced subgraph. Many vertex-deletion problems such as \textsc{Perfect Deletion}, \textsc{Wheel-free Deletion}, and \textsc{Interval Deletion} fit into this framework. We introduce the concept of \emph{characterizing a graph property~$\Pi$ by low-rank adjacencies}, and use it as the cornerstone of a general kernelization theorem for \textsc{$\Pi$-Free Deletion} parameterized by the size of a vertex cover. The resulting framework captures problems such as \textsc{AT-Free Deletion},  \textsc{Wheel-free Deletion}, and \textsc{Interval Deletion}. Moreover, our new framework shows that the vertex-deletion problem to perfect graphs has a polynomial kernel when parameterized by vertex cover, thereby resolving an open question by Fomin et al.~[JCSS 2014]. Our main technical contribution shows how linear-algebraic dependence of suitably defined vectors over~$\mathbb{F}_2$ implies graph-theoretic statements about the presence of forbidden induced subgraphs.
\end{abstract}

\clearpage

\section{Introduction}\label{sec:introduction}
\subparagraph{Background}
This paper continues a long line of investigation~\cite{BodlaenderFLPST16,BodlaenderCrossComposition,FOMIN2014468,GiannopoulouForbiddenMinors,JansenP19,KernelizationBook2018}, aimed at answering the following question: how and when can an efficient preprocessing algorithm reduce the size of inputs to NP-hard problems, without changing their answers? This question can be framed and answered using the notion of kernelization, which originated in parameterized complexity theory. 

In parameterized complexity theory, the complexity analysis is done not only in the size of the input, but also in terms of another complexity measure related to the input. This complexity measure is called the \emph{parameter}. For graph problems, typical parameters are the size of a solution, the treewidth of the graph, or the size of a minimum vertex cover (the \emph{vertex cover number}). The latter two are often called structural parameterizations. A kernelization is a polynomial-time preprocessing algorithm with a performance guarantee. It reduces an instance $(x,k)$ of a parameterized problem to an instance $(x',k')$ that has an equivalent \textsc{yes}/\textsc{no} answer, such that $|x'|$ and $k'$ are bounded by $f(k)$ for some computable function $f$, called the size of the kernel. If $f$ is a polynomial function, the parameterized problem is said to admit a polynomial kernel. Polynomial kernels are highly sought after, as they allow problem instances to be reduced to a relatively small size.

We investigate polynomial kernels for the class of \emph{graph modification} problems, in an attempt to develop a widely applicable and generic kernelization framework. In graph modification problems, the goal is to make a small number of changes to an input graph to make it satisfy a certain property. Possible modifications are vertex deletions, edge deletions, and edge additions. In this work, we consider the problem of deleting a  bounded-size set of vertices such that the resulting graph does not contain certain graphs as an induced subgraph. 

The study of kernelization for graph modification problems parameterized by solution size has an interesting and rich history~\cite{AGRAWAL2019,BurrageEFLMR06,CaoRSY18,CyganPPLW17,FominLMS12,Iwata17,JansenP18,KratschW14}. However, some graph modification problems such as \textsc{Perfect Vertex Deletion}~\cite{Heggernes:2013:PCV:2560981.2561653} and \textsc{Wheel-free Vertex Deletion}~\cite{Lokshtanov08} are W[2]-hard parameterized by the solution size and therefore do not admit any kernels unless FPT = W[2]. Together with the intrinsic interest in obtaining generic kernelization theorems that apply to a large class of problems with a single parameter, this has triggered research into polynomial kernelization for graph problems under structural parameterizations~\cite{BougeretS19,FOMIN2014468,GiannopoulouForbiddenMinors,JansenP19,UhlmannW13} such as the vertex cover number. The latter parameter is often used for its mathematical elegance, and due to the fact that slightly less restrictive parameters such as the feedback vertex number already cause simple problems such as \textsc{3-Coloring} not to admit polynomial kernels~\cite{JansenK13}, under the standard assumption \notcontainment. This work therefore focuses on the following class of NP-hard~\cite{LEWIS1980219} parameterized problems, where~$\Pi$ is a fixed (possibly infinite) set of graphs:
\par
\defparproblem{$\Pi$\textsc{-free Deletion}}{A graph $G$, a vertex cover $X$ of $G$, and an integer $k$.}{$|X|$}{Does there exist a set $S \subseteq V(G)$ of size at most $k$ such that $G - S$ does not contain any graph from $\Pi$ as induced subgraph?}
\newline
The assumption that a vertex cover~$X$ is given in the input is for technical reasons. If the problem would be parameterized by an upper-bound on the vertex cover number of the graph, without giving such a vertex cover, then the kernelization algorithm would have to verify that this is indeed a correct upper bound; an NP-hard problem. Instead, in this setting we just want to allow the kernelization algorithm to exploit the structural restriction guaranteed by having a small vertex cover in the graph. We refer to the discussion by Fellows et al.~\cite[\S 2.2]{FellowsJR13} for more background. To apply the kernelization algorithms for problems defined in this way, one may simply use a $2$-approximate vertex cover as~$X$.

Fomin et al.~\cite{FOMIN2014468} have investigated characteristics of $\Pi$-\textsc{free Deletion} problems that admit a polynomial kernel parameterized by the size of a vertex cover. They introduced a generic framework that poses three conditions on the graph property $\Pi$, which are sufficient to reach a polynomial kernel for $\Pi$-\textsc{free Deletion} parameterized by vertex cover. Examples of graph properties that fit in their framework are for instance `having a chordless cycle of length at least 4' or `having an odd cycle'. This results in polynomial kernels for \textsc{Chordal Deletion} and \textsc{Odd Cycle Transversal} respectively.
\textsc{Interval Deletion} does not fit in this framework, even though interval graphs are hereditary. Agrawal et al.~\cite{AGRAWAL2019} show that it admits a polynomial kernel parameterized by solution size, and therefore also by vertex cover size. They introduced a linear-algebraic technique, which assigns a vector over~$\mathbb{F}_2$ to each vertex, to find an induced subgraph that preserves the size of an optimal solution by combining several disjoint bases of systems of such vectors. This formed the inspiration for our work, in which we improve the generic kernelization framework of Fomin et al.~\cite{FOMIN2014468} using linear-algebraic techniques inspired by the kernel~\cite{AGRAWAL2019} for \textsc{Interval Deletion}.

\subparagraph{Results}

\begin{table}
\centering
\begin{tabular}{l c}
\hline
Problem & Vertices in kernel \\
\hline 
\textsc{Perfect Deletion} & $\mathcal{O}(|X|^5)$\\
\textsc{Even-hole-free Deletion} ($\bigstar$) & $\mathcal{O}(|X|^4)$\\
\textsc{AT-free Deletion} & $\mathcal{O}(|X|^9)$\\
\textsc{Interval Deletion} & $\mathcal{O}(|X|^9)$ \\
\textsc{Wheel-free Deletion} & $\mathcal{O}(|X|^5)$\\
\hline 
\end{tabular}
\caption{Kernels obtained by our framework for problems parameterized by a vertex cover $X$.}
\label{table:results}
\end{table}

We introduce the notion of \emph{characterizing a graph property~$\Pi$ by low-rank adjacencies}, and use it to generalize the kernelization framework by Fomin et al.~\cite{FOMIN2014468} significantly. The resulting kernelization algorithms consist of a single, conceptually simple reduction rule for \textsc{$\Pi$-free Deletion}, whose property-specific correctness proofs show how the linear dependence of suitably defined vectors implies certain graph-theoretic properties. This results in a simpler kernelization for \textsc{Interval Deletion} parameterized by vertex cover compared to the one by Agrawal et al.~\cite{AGRAWAL2019}. More importantly, several vertex-deletion problems whose kernelization complexity was previously open can be covered by the framework. These include \textsc{AT-free Deletion} (eliminate all asteroidal triples~\cite{KOHLER_ATFREE} from the graph), \textsc{Wheel-free Deletion}, and also \textsc{Perfect Deletion} which was an explicit open question of Fomin et al.~\cite[\S 5]{FOMIN2014468}. An overview is given in \cref{table:results}. Moreover, we give evidence that the distinguishing property of our framework (being able to characterize~$\Pi$ by low-rank adjacencies) is the right one to capture kernelization complexity. While the \textsc{Wheel-free Deletion} problem fits into our framework and therefore has a polynomial kernel, the situation is very different for the related problem \textsc{Almost Wheel-free Deletion} (ensure the resulting graph does not contain any wheel, except possibly~$W_4$). We prove the latter problem does not fit into our framework, and that it does not admit a polynomial kernel parameterized by vertex cover, unless \containment.

\subparagraph{Related work}
Even though the vertex cover is generally not small compared to the size of the input graph, it is not always the case that a polynomial kernel parameterized by vertex cover number exists. This was shown by Bodlaender et al.~\cite{BodlaenderCrossComposition}. They showed that for instance the \textsc{Clique} problem that asks whether a graph contains a clique of $k$ vertices, does not admit a polynomial kernel parameterized by the vertex cover size, unless coNP $\subseteq$ NP/poly.
\par
A graph is perfect if for every induced subgraph $H$, the chromatic
number of $H$ is equal to the size of the largest clique of $H$. Conjectured by Berge in 1961 and proven in the beginning of this century by Chudnovsky et al.~\cite{zbMATH05081753}, the strong perfect graph theorem states that a graph is perfect if and only if it is Berge. The forbidden induced subgraphs of Berge graphs (and hence of perfect graphs) are $C_{2k+1}$ and $\overline{C}_{2k+1}$ for $k \geq 2$, that is, induced cycles and their edge complements of odd length at least 5. 
A survey of forbidden subgraph characterizations of some other hereditary graph classes is given in~\cite[Chapter 7]{GraphClassesSurvey}. 


\subparagraph{Organization}
In \cref{sec:preliminaries} we give preliminaries and definitions used throughout this work. In \cref{sec:framework} we introduce the framework. In \cref{sec:using} we show that several problems such as \textsc{Perfect Deletion} and \textsc{Interval Deletion} fit in this framework. Finally we conclude in \cref{sec:conclusion}. For statements marked $\bigstar$, the proof can be found in the appendix.

\section{Preliminaries}\label{sec:preliminaries}
\subparagraph{Notation}
For $i \in \mathbb{N}$, we denote the set $\{1,...,i\}$ by $[i]$. For a set $S$, we denote the set of subsets of size at most $k$ by $\binom{S}{\leq k} = \{S' \subseteq S \mid |S'| \leq k\}$. Similarly,~$\binom{S}{k}$ denotes the set of subsets of size exactly~$k$. 
We consider simple graphs that are unweighted and undirected without self-loops. A graph $G$ has vertex and edge sets $V(G)$ and $E(G)$ respectively. An edge between vertices $u,v \in V(G)$ is an unordered pair $\{u,v\}$. 
For a set of vertices $S \subseteq V(G)$, by $G[S]$ we denote the graph induced by $S$. For $v \in V(G)$ and $S \subseteq V(G)$, by $G - v$ and $G - S$ we mean the graphs $G[V(G) \setminus \{v\}]$ and $G[V(G) \setminus S]$ respectively. 
We denote the open neighborhood of $v \in V(G)$ by $N_G(v) = \{u \mid \{u,v\} \in E(G)\}$. When clear from context, we sometimes omit the subscript $G$.
For a graph $G$, let $\overline{G}$ be the edge complement graph of $G$ on the same vertex set, such that for distinct $u,v \in V(G)$ we have $\{u,v\} \in E(\overline{G})$ if and only if $\{u,v\} \notin E(G)$. 
The path graph on~$n$ vertices~$(v_1,...,v_n)$ is denoted by $P_n$. Similarly, the $n$-vertex cycle for~$n \geq 3$ is denoted by $C_n$. When $n \geq 4$, the graph $C_n$ is often called a \emph{hole}. 
For $n \geq 3$, the wheel $W_n$ of size $n$ is the graph on vertices $\{c,v_1,...,v_n\}$ such that $(v_1,...,v_n)$ is a cycle and $c$ is adjacent to $v_i$ for all $i \in [n]$. 
An asteroidal triple (AT) in a graph $G$ consists of three vertices such that every pair is connected by a path that avoids the neighborhood of the third. A vertex cover in a graph~$G$ is a set of vertices that contains at least one endpoint of every edge. The minimum size of a vertex cover in a graph $G$ is denoted by $\vc(G)$.

\subparagraph{Parameterized complexity}
A parameterized problem~\cite{CYGANPARALG2015,DowneyF13} is a language $Q \subseteq \Sigma^* \times \mathbb{N}$, where~$\Sigma$ is a finite alphabet. The notion of kernelization is formalized as follows.

\begin{definition}
Let $Q \subseteq \Sigma^* \times \mathbb{N}$ be a parameterized problem and let~$f \colon \mathbb{N} \to \mathbb{N}$ be a computable function. A \emph{kernelization} for~$Q$ of size~$f$ is an algorithm that, given an instance $(x,k) \in \Sigma^* \times \mathbb{N}$, outputs in time polynomial in $|x| + k$ an instance $(x',k')$ (known as the kernel) such that $(x,k) \in Q$ if and only if $(x',k') \in Q$ and such that $|x'|,k' \leq f(k)$. If $f$ is a polynomial function, then the algorithm is a \emph{polynomial kernelization}.
\end{definition}

\subparagraph{Previous kernelization framework}
We state some of the results from the kernelization framework by Fomin et al.~\cite{FOMIN2014468} that forms the basis of this work. A graph property $\Pi$ is a (possibly infinite) set of graphs.

\begin{definition}[Definition 3, \cite{FOMIN2014468}]\label{def:characterization}
A graph property $\Pi$ is characterized by $c_{\Pi} \in \mathbb{N}$ adjacencies if for all graphs $G \in \Pi$, for every vertex $v \in V(G)$, there is a set $D \subseteq V(G) \setminus \{v\}$ of size at most $c_{\Pi}$ such that all graphs $G'$ which are obtained from $G$ by adding or removing edges between $v$ and vertices in $V(G)\setminus D$, are also contained in $\Pi$. 
\end{definition}
As an example, the graph property `having a chordless cycle of length at least 4' is characterized by 3 adjacencies. The graph property `not being an interval graph' is not characterized by a finite number of adjacencies. Other examples are given by Fomin et al.~\cite{FOMIN2014468}. 
\par
Any finite graph property $\Pi$ is trivially characterized by $\max_{G\in \Pi} |V(G)|-1$ adjacencies. We state the following easily verified fact without proof.
\begin{proposition} \label{prop:infiniteset_finitegraphs}
Let $\Pi'$ be the set of all graphs that contain a graph from a finite set $\Pi$ as induced subgraph. Then $\Pi'$ is characterized by $\max_{G\in \Pi} |V(G)|-1$ adjacencies.
\end{proposition}
\par
A graph $G$ is vertex-minimal with respect to $\Pi$ if $G \in \Pi$ and for all $S \subsetneq V(G)$ the graph $G[S]$ is not contained in $\Pi$. The following framework can be used to get polynomial kernels for the \textsc{$\Pi$-free Deletion} problem parameterized by vertex cover.

\begin{theorem}[Theorem 2, \cite{FOMIN2014468}]\label{thm:oldframework}
If $\Pi$ is a graph property such that:
\begin{romanenumerate}
    \item $\Pi$ is characterized by $c_{\Pi}$ adjacencies,
    \item every graph in $\Pi$ contains at least one edge, and
    \item there is a non-decreasing polynomial $p \colon \mathbb{N} \rightarrow \mathbb{N}$ such that all graphs $G$ that are vertex-minimal with respect to $\Pi$ satisfy $|V(G)| \leq p(\vc(G))$,
\end{romanenumerate}
then $\Pi$-\textsc{free Deletion} parameterized by the vertex cover size $x$ admits a polynomial kernel with $\mathcal{O}((x+p(x))x^{c_{\Pi}})$ vertices.
\end{theorem}

\section{Framework based on low-rank adjacencies}\label{sec:framework}
\subsection{Incidence vectors and characterizations}
As a first step towards our kernelization framework for \textsc{$\Pi$-free Deletion}, we introduce an incidence vector definition ($\incd$) that characterizes the neighborhood of a given vertex. Compared to the vector encoding used by Agrawal et al.~\cite{AGRAWAL2019} for \textsc{Interval Deletion}, our vector definition differs because it supports arbitrarily large subsets (they consider subsets of size at most two), and because an entry of a vector simultaneously prescribes which neighbors should be present, and which neighbors should \emph{not} be present.

\begin{definition}[$c$-incidence vector]\label{def:inc}
Let $G$ be a graph with vertex cover $X$ and let $c \in \mathbb{N}$. Let $Q',R' \subseteq X$ such that $|Q'| + |R'| \leq c$. We define the \emph{$c$-incidence vector} $\incf{c}{G}{X}{Q'}{R'}(u)$ for a vertex $u \in V(G) \setminus X$ as a vector over $\mathbb{F}_2$ that has an entry for each $(Q,R) \in X \times X$ with $Q \cap R = \emptyset$ such that $|Q|+|R| \leq c$, $Q' \subseteq Q$ and $R' \subseteq R$. It is defined as follows: 

\begin{equation*}
\incf{c}{G}{X}{Q'}{R'}(u)[Q,R] =
    \begin{cases}
    1 & \text{if $N_G(u) \cap Q = \emptyset$ and $R \subseteq N_G(u)$,}\\
    0 & \text{otherwise.}
    \end{cases}
\end{equation*}

\end{definition}

We drop superscript $(Q',R')$ if both $Q'$ and $R'$ are empty sets. The intuition behind the superscript $(Q',R')$ is that it projects the entries of the full incidence vector~$\incs{c}{G}{X}$ to those for supersets of~$Q', R'$. The $c$-incidence vectors can be naturally summed coordinate-wise. For ease of presentation we do not define an explicit order on the coordinates of the vector, as any arbitrary but fixed ordering suffices. 

If the sum of some vectors equals some other vector with respect to a certain graph $G$, then this equality is preserved when decreasing $c$ or taking induced subgraphs of $G$.

\begin{proposition}\label{note:incsubgraph}\label{note:largerrank}
Let $G$ be a graph with vertex cover $X$, let $c \in \mathbb{N}$, and let $D \subseteq V(G)$ be disjoint from $X$. If $v \in V(G) \setminus(D \cup X)$ and $\incs{c}{G}{X} (v) = \sum_{u \in D} \incs{c}{G}{X} (u)$, then
\begin{itemize}
    \item $\incs{c'}{G}{X} (v) = \sum_{u \in D} \incs{c'}{G}{X} (u)$ for any $c' \leq c$, and
    \item $\incs{c}{H}{X \cap V(H)} (v) = \sum_{u \in D} \incs{c}{H}{X \cap V(H)} (u)$ for any induced subgraph $H$ of $G$ that contains $D$ and $v$.
\end{itemize}
\end{proposition}
\begin{proof}
For the first point, observe that for any vertex~$v \notin (D \cup X)$, the vector $\incs{c'}{G}{X}(v)$ is simply a projection of $\incs{c}{G}{X}(v)$ to a subset of its coordinates. Hence if the complete vector of~$v$ is equal to the sum of the complete vectors of~$u \in D$, then projecting the vector of both~$v$ and of the sum to the same set of coordinates, yields identical vectors.

For the second point, observe that since~$X$ is a vertex cover of~$G$, we have~$N_G(v) \subseteq X$ for all~$v \in V(G) \setminus X$. Moreover, if~$H$ is an induced subgraph of~$G$ containing~$D$ and~$v$, then~$X_H := X \cap V(H)$ is a vertex cover of~$H$. Hence for any~$u \in V(H) \setminus X_H$ the $c$-incidence vector $\incs{c}{H}{X \cap V(H)}(u)$ is well-defined. If~$Q,R$ are disjoint sets for which $\incs{c}{H}{X \cap V(H)}(u)[Q,R]$ is defined, then~$Q,R \subseteq X_H$, so the adjacencies between~$u$ and~$Q \cup R$ in the induced subgraph~$H$ are identical to those in~$G$, which implies~$\incs{c}{G}{X}(u)[Q,R] = \incs{c}{H}{X \cap V(H)}(u)[Q,R]$. Hence when we replace a $c$-incidence vector with subscript~$(G,X)$ by a vector with subscript $(H,X \cap V(H))$, we essentially project the vector to a subset of its coordinates without changing any values. For the same reason as above, this preserves the fact that the vectors of~$D$ sum to that of~$v$.
\end{proof}

We are ready to introduce the main definition, namely characterization of a graph property $\Pi$ by rank-$c$ adjacencies for some $c \in \mathbb{N}$. In our framework, this replaces characterization by $c$ adjacencies in the framework of Fomin et al.~\cite{FOMIN2014468} (\cref{thm:oldframework}). 

\begin{definition}[rank-$c$ adjacencies]\label{def:rankc}
Let $c \in \mathbb{N}$ be a natural number. Graph property $\Pi$ is characterized by rank-$c$ adjacencies if the following holds. For each graph $H$, for each vertex cover $X$ of $H$, for each set $D \subseteq V(H) \setminus X$, for each $v \in V(H) \setminus (D \cup X)$, if 
\begin{itemize}
    \item $H - D \in \Pi$, and
    \item $\incs{c}{H}{X} (v) = \sum_{u \in D} \incs{c}{H}{X} (u)$ when evaluated over~$\mathbb{F}_2$,
\end{itemize}
then there exists $D' \subseteq D$ such that $H - v - (D \setminus D') \in \Pi$. If there always exists such set $D'$ of size 1, then we say $\Pi$ is characterized by rank-$c$ adjacencies with \emph{singleton replacements}.

\end{definition}
Intuitively, the definition demands that if we have a set $D$ such that $H-D \in \Pi$, and the $c$-incidence vectors of $D$ sum to the vector of some vertex $v$ over $\mathbb{F}_2$, then there exists $D' \subseteq D$ such that removing $v$ from $H-D$ and adding back $D'$ results in a graph that is still contained in $\Pi$. For example, in \cref{sec:perfect} we show that the graph property `containing an odd hole or odd-anti-hole' is characterized by rank-4 adjacencies. Using our framework, this leads to a polynomial kernel for \textsc{Perfect Deletion} parameterized by vertex cover. Other examples of graph properties which are characterized by a rank-$c$ adjacencies for some~$c \in \mathcal{O}(1)$ include `containing a cycle' and `being wheel-free'. On the other hand, we will show in Theorem~\ref{thm:awfnorankc} that the property `containing an induced wheel whose size is 3 or at least 5' cannot be characterized by rank-$c$ adjacencies for any finite~$c$.

\subsection{A generic kernelization}
Our kernelization framework for \textsc{$\Pi$-free Deletion} relies on a single reduction rule presented in \cref{alg:reduce}. It assigns an incidence vector to every vertex outside the vertex cover and uses linear algebra to select vertices to store in the kernel. Let us therefore recall the relevant algebraic background. A \emph{basis} of a set~$S$ of $d$-dimensional vectors over a field~$\mathbb{F}$ is a minimum-size subset~$B \subseteq S$ such that all~$\vec{v} \in S$ can be expressed as linear combinations of elements of~$B$, i.e.,~$\vec{v} = \sum _{\vec{u} \in B} \alpha_{\vec{u}} \cdot \vec{u}$ for a suitable choice of coefficients~$\alpha_{\vec{u}} \in \mathbb{F}$. When working over the field~$\mathbb{F}_2$, the only possible coefficients are~$0$ and~$1$, which gives a basis~$B$ of~$S$ the stronger property that any vector~$\vec{v} \in S$ can be written as~$\sum_{\vec{u} \in B'} \vec{u}$, where~$B' \subseteq B$ consists of those vectors which get a coefficient of~$1$ in the linear combination.

Our reduction algorithm repeatedly computes a basis of the incidence vectors of the remaining set of vertices, and stores the vertices corresponding to the basis in the kernel.

\begin{algorithm}
\caption{Reduce (Graph $G$, vertex cover $X$ of $G$, $\ell \in \mathbb{N}$, $c \in \mathbb{N}$)}
\label{alg:reduce}
\begin{algorithmic}[1]
\STATE Let $Y_1 := V(G) \setminus X$.
\FOR{$i \leftarrow 1$ to $\ell$}
\STATE Let $V_i = \{\incs{c}{G}{X}(y) \mid y \in Y_i\}$ and compute a basis $B_i$ of $V_i$ over $\mathbb{F}_2$. \label{alg:reduce:basisline}
\STATE For each $\vec{v} \in B_i$, choose a unique vertex $y_{\vec{v}} \in Y_i$ such that $\vec{v} = \incs{c}{G}{X}(y_{\vec{v}})$.
\STATE Let $A_i := \{y_{\vec{v}} \mid \vec{v} \in B_i\}$ and $Y_{i+1} = Y_i \setminus A_i$.
\ENDFOR
\RETURN $G[X \cup \bigcup_{i=1}^\ell A_i]$
\end{algorithmic}
\end{algorithm}

\begin{proposition}\label{prop:reducetimegraphsize}
For a fixed $c \in \mathbb{N}$, \cref{alg:reduce} runs in polynomial time in terms of $\ell$ and the size of the graph, and returns a graph on $\mathcal{O}(|X| + \ell \cdot |X|^c)$ vertices.
\end{proposition}
\begin{proof}
Observe that for each $i$, the vectors in $V_i$ have at most $2^c \cdot |\binom{X}{\leq c}| = \mathcal{O} (|X|^c)$ entries and therefore the rank of the vector space is $\mathcal{O}(|X|^c)$. Hence each computed basis contains $\mathcal{O}(|X|^c)$ vectors. For constant $c$, this means that each basis can be computed in polynomial time using Gaussian elimination. The remaining operations can be done in polynomial time in terms of $\ell$ and the size of the graph. Since~$|A_i| \in \mathcal{O}(|X|^c)$ for each~$i \in [\ell]$, the resulting graph has~$\mathcal{O}(|X| + \ell \cdot |X|^c)$ vertices.
\end{proof}

\begin{theorem}\label{thm:framework}
If $\Pi$ is a graph property such that:
\begin{romanenumerate}
\item $\Pi$ is characterized by rank-$c$ adjacencies,
\item \label{enum:framework2} every graph in $\Pi$ contains at least one edge, and
\item \label{enum:framework3}there is a non-decreasing polynomial $p: \mathbb{N} \rightarrow \mathbb{N}$ such that all graphs $G$ that are vertex-minimal with respect to $\Pi$ satisfy $|V(G)| \leq p(\vc(G))$,
\end{romanenumerate}
then \textsc{$\Pi$-free Deletion} parameterized by the vertex cover size $x$ admits a polynomial kernel on $\mathcal{O}((x+p(x))\cdot x^c)$ vertices.
\end{theorem}
\begin{proof}
Consider an instance $(G,X,k)$ of \textsc{$\Pi$-free Deletion}. Note that if $k \geq |X|$, then we can delete the entire vertex cover to get an edgeless graph, which is $\Pi$-free by \eqref{enum:framework2}, and therefore we may output a constant size \textsc{yes}-instance as the kernel. If $k < |X|$, let $G'$ be the graph obtained by the procedure \textsc{Reduce}($G$,$X$,$\ell := k+1+p(|X|)$,$c$). By \cref{prop:reducetimegraphsize} this can be done in polynomial time and the resulting graph contains $\mathcal{O}((|X|+p(|X|))\cdot |X|^c)$ vertices. All that is left to show is that the instance $(G',X,k)$ is equivalent to the original instance.  Since $G'$ is an induced subgraph of $G$, it follows that if $(G,X,k)$ is a \textsc{yes}-instance, then so is $(G',X,k)$. In the other direction, suppose that $(G',X,k)$ is a \textsc{yes}-instance with solution $S$. We show that $S$ also is a solution for the original instance.

For the sake of contradiction assume that this is not the case. Then the graph $G-S$ contains an induced subgraph that belongs to $\Pi$. Let $P$ be a minimal set of vertices of $G-S$ for which $G[P] \in \Pi$ and that minimizes $|P \setminus V(G')|$. 
Since $S$ is a solution for $(G',X,k)$, it follows that there exists a vertex $v \in P \setminus V(G')$. Moreover we have that $v \notin X$, since the graph~$G'$ returned by \cref{alg:reduce} contains all vertices of~$X$. 
The set $P \cap X$ is a vertex cover for $G[P]$, therefore by property \eqref{enum:framework3} we have that $|P| \leq p(\vc(G[P])) \leq p(|X|)$. Since the vertex sets~$A_1, \ldots, A_\ell$ computed in the \textsc{Reduce} operation are disjoint, and since $|S| \leq k$, it follows that there exists an $i \in [k + 1 + p(|X|)]$ such that the set of vertices $A_i$ corresponding to basis $B_i$ is disjoint from both $S$ and $P$.

As $v \notin V(G')$ implies~$v \notin \bigcup_{i=1}^\ell A_i$, in each iteration of line~\ref{alg:reduce:basisline} the vectors of the computed vertex set~$A_i$ span the vector of~$v$. 
Hence, since we work over $\mathbb{F}_2$, there exists $D \subseteq A_i \subseteq V(G')$ such that $\incs{c}{G}{X}(v) = \sum_{u \in D} \incs{c}{G}{X}(u)$.
Consider the graph $H := G[P \cup D]$. Since $H$ is an induced subgraph that includes $D$ and $D$ is disjoint from $X$, by \cref{note:incsubgraph} it follows that $\incs{c}{H}{X \cap V(H)}(v) = \sum_{u \in D} \incs{c}{H}{X \cap V(H)}(u)$. Moreover $H-D \in \Pi$ as $H-D=G[P]$.
By the definition of rank-$c$ adjacencies it follows that there exists $D' \subseteq D$ such that $P' = H - v - (D \setminus D') \in \Pi$. But since $|P'\setminus V(G')| < |P \setminus V(G')|$, this contradicts the minimality of $P$. Therefore $S$ must be a solution for the original instance.
\end{proof}

\subsection{Properties of low-rank adjacencies}
In this section we present several technical lemmata dealing with low-rank adjacencies. These will be useful when applying the framework to various graph properties. The next lemma shows that if $\Pi$ is characterized by low-rank adjacencies with singleton replacements, then the edge-complement graphs are as well.

\begin{lemma}\label{lemma:complementgraphprop}
Let $\Pi$ be a graph property that is characterized by rank-$c$ adjacencies with singleton replacements. Let $\overline{\Pi}$ be the graph property such that $G \in \Pi$ if and only if $\overline{G} \in \overline{\Pi}$. Then $\overline{\Pi}$ is characterized by rank-$c$ adjacencies with singleton replacements.
\end{lemma}
\begin{proof}
Let $H$ be a graph with vertex cover $X$. Let $D \subseteq V(H) \setminus X$ be a set such that $H-D \in \overline{\Pi}$. Consider some vertex $v \in V(H) \setminus (D \cup X)$ such that $\incs{c}{H}{X} (v) = \sum_{u \in D} \incs{c}{H}{X} (u)$. Let $X' = V(H) \setminus (D \cup \{v\})$. 

\begin{claim} \label{claim:complement:sums}
We have $\incs{c}{H}{X'} (v) = \sum_{u \in D} \incs{c}{H}{X'} (u)$.
\end{claim}
\begin{claimproof}
Since vertices outside $X$ are independent, neither $v$ nor any vertex in $D$ is adjacent to any vertex in $X' \setminus X$. So for any disjoint~$Q,R \subseteq X'$ with~$R \cap (X' \setminus X) \neq \emptyset$ we have~$\incs{c}{H}{X'}(v)[Q,R] = \incs{c}{H}{X'}(u)[Q,R] = 0$ for all~$u \in D$ by definition, while for~$R \cap (X' \setminus X) = \emptyset$ we have~$\incs{c}{H}{X'}(u)[Q,R] = \incs{c}{H}{X}(u)[Q \cap X,R]$ for any~$u \in D \cup \{v\}$.
\end{claimproof}
Let $H'$ be obtained from $H$ by (1) taking the edge complement, and then (2) turning $H'[D \cup \{v\}]$ back into an independent set (the complement made it a clique). Note that $X'$ is a vertex cover of $H'$.

\begin{claim}
We have $\incs{c}{H'}{X'} (v) = \sum_{u \in D} \incs{c}{H'}{X'} (u)$.
\end{claim}
\begin{claimproof}
Immediate from \cref{claim:complement:sums} since  $\incs{c}{H'}{X'}(u)[Q,R] = \incs{c}{H}{X'}(u)[R,Q]$ for all~$u \in D \cup \{v\}$.
\end{claimproof}

Observe that~$H' - D$ is the edge-complement of~$H - D$, so~$H'- D \in \Pi$. Together with the previous claim, since $\Pi$ is characterized by rank-$c$ adjacencies with singleton replacements, it follows that there exists $v' \in D$ such that $G' := H'- v - (D \setminus \{v'\}) \in \Pi$. Since~$G'$ contains only a single vertex of~$\{v\} \cup D$, none of its edges were edited during step (2) above, so that~$G := H - v - (D \setminus \{v'\})$ is the edge-complement of~$G'$, implying~$G \in \overline{\Pi}$. This shows that $\overline{\Pi}$ is characterized by rank-$c$ adjacencies with singleton replacements.
\end{proof}

\cref{prop:graphpropertymerge} proves closure under taking the union of two characterized properties. 

\begin{lemma}\label{prop:graphpropertymerge}
Let $\Pi$ and $\Pi'$ be graph properties characterized by rank-$c_{\Pi}$ and rank-$c_{\Pi'}$ adjacencies (with singleton replacements), respectively. Then the property $\Pi \cup \Pi'$ is characterized by rank-$\max(c_{\Pi},c_{\Pi'})$ adjacencies (with singleton replacements). 
\end{lemma}
\begin{proof}
Consider a graph $H$ with vertex cover $X$ and set $D\subseteq V(H) \setminus X$ such that $H-D \in \Pi \cup \Pi'$. Let $v \in V(H) \setminus (D \cup X)$ be some vertex such that $\incs{\max(c_\Pi,c_{\Pi'})}{H}{X} (v) = \sum_{u \in D} \incs{\max(c_\Pi,c_{\Pi'})}{H}{X} (u)$. By \cref{note:largerrank}, we have $\incs{c_\Pi}{H}{X} (v) = \sum_{u \in D} \incs{c_\Pi}{H}{X} (u)$. If $H-D \in \Pi$, then there exists $D' \subseteq D$ such that $H-v-(D\setminus D') \in \Pi$ and hence, $H-v-(D\setminus D') \in \Pi \cup \Pi'$ (in case of singleton replacements, $D'$ is replaced by $\{v'\}$ for some $v' \in D$). The case~$H - D \in \Pi'$ is symmetric.
\end{proof}

While the \emph{intersection} of two graph properties which are characterized by a \emph{finite number} of adjacencies is again characterized by a finite number of adjacencies~\cite[Proposition 4]{FOMIN2014468}, the same does not hold for low-rank adjacencies; there is no analog of \cref{prop:graphpropertymerge} for intersections.


In a graph $G$, we say that vertices $u$ and $v$ \emph{share adjacencies} to a set $S$, if $N_G(u) \cap S = N_G(v) \cap S$. The following lemma states that when we have a set $D$ whose $c$-incidence vectors sum to the vector of $v$, then for any set~$S$ of size up to~$c$ there exists a nonempty subset $D'\subseteq D$ whose members all share adjacencies with $v$ to~$S$.

\begin{lemma}\label{prop:adjacencyshare}
Let $G$ be a graph with vertex cover $X$, let $D \subseteq V(G)$ be disjoint from $X$, and let~$c \in \mathbb{N}$. Consider a vertex $v \in V(G) \setminus(D \cup X)$. If $\incs{c}{G}{X} (v) = \sum_{u \in D} \incs{c}{G}{X} (u)$, then for any set $S \subseteq V(G)$ with $|S|\leq c$ there exists $D' \subseteq D$, such that:
\begin{itemize}
    \item $|D'| \geq 1$ is odd,
    \item each vertex $u \in D'$ shares adjacencies with $v$ to $S$, and
    \item $\incf{c}{G}{X}{Q'}{R'} (v) = \sum_{u \in D'} \incf{c}{G}{X}{Q'}{R'} (u)$, where $Q' = (S \setminus N_G(v)) \cap X$ and $R' = S \cap N_G(v)$.
\end{itemize}
\end{lemma}
\begin{proof}
For any vertex $d \in D$ that does not share adjacencies with $v$ to $S$, the vector $\incf{c}{G}{X}{Q'}{R'} (d)$ is the vector containing only zeros. Let $D' \subseteq D$ be the set of vertices that do share adjacencies with $v$ to $S$. Clearly $\incf{c}{G}{X}{Q'}{R'} (v) = \sum_{u \in D'} \incf{c}{G}{X}{Q'}{R'} (u)$, as removing all-zero vectors does not change the sum. Since $\incf{c}{G}{X}{Q'}{R'} (v)[Q',R'] =1$ and $\incf{c}{G}{X}{Q'}{R'} (u)[Q',R'] =1$ for all $u \in D'$, $|D'| \geq 1$ must be odd.
\end{proof}

Our framework adapts \cref{thm:oldframework} by replacing characterization by $c$ adjacencies by rank-$c$ adjacencies. From the following statement we can conclude that our framework extends \cref{thm:oldframework}.

\begin{lemma}\label{lemma:c_to_rankc}
A graph property $\Pi$ characterized by $c$ adjacencies is also characterized by rank-$c$ adjacencies with singleton replacements.
\end{lemma}
\begin{proof}
Let $\Pi$ be a graph property characterized by $c$ adjacencies. We show that $\Pi$ is characterized by rank-$c$ adjacencies. Let $G$ be a graph with vertex cover $X$ and $D \subseteq V(G) \setminus X$ be a set such that $G-D \in \Pi$. Let $v \in V(G) \setminus (D \cup X)$ be a vertex such that $\incs{c}{G}{X}(v) = \sum_{u \in D} \incs{c}{G}{X}(u)$.

Since $\Pi$ is characterized by $c$ adjacencies, there exists a set $B$ of size at most $c$ such that all graphs obtained by changing adjacencies between $v$ and $V(G) \setminus B$ are also contained in~$\Pi$. By \cref{prop:adjacencyshare} there exists $w \in D$ that shares adjacencies with $v$ to $B$. Now consider the graph $G-v-(D\setminus \{w\})$. This graph is isomorphic to $G-D$ where $w$ is matched to $v$ and the adjacencies between $v$ and $V(G) \setminus B$ are changed. But then by the definition of characterization by $c$ adjacencies it follows that $G-v-(D\setminus \{w\}) \in \Pi$. 
\end{proof}

\section{Using the framework}\label{sec:using}

In this section we give some results using our framework, which are listed in \cref{table:results}. We give polynomial kernels for \textsc{Perfect Deletion}, \textsc{AT-free Deletion}, \textsc{Interval Deletion}, \textsc{Even-hole-free Deletion}, and \textsc{Wheel-free Deletion} parameterized by vertex cover.  

\subsection{Perfect Deletion}\label{sec:perfect}
Let $\Pi_P$ be the set of graphs that contain an odd hole or an odd anti-hole. The \textsc{$\Pi_P$-free Deletion} problem is known as the \textsc{Perfect Deletion} problem. 
It was mentioned as an open question by Fomin et al.~\cite{FOMIN2014468}, since one can show that $\Pi_P$ is not characterized by a finite number of adjacencies. In this section we show that $\Pi_P$ is characterized by rank-4 adjacencies with singleton replacements. 
Following this result, we show that it admits a polynomial kernel using \cref{thm:framework}. First, we give a lemma that will be helpful in the proof later on. We say that a vertex \emph{sees an edge} if it is adjacent to both of its endpoints.

\begin{lemma}\label{lemma:odd(anti)hole_(odd)evenseen}
Let $G$ be a graph, $P = (v_1,...,v_n)$ where $n \geq 4$ is even be an induced path in $G$, and let $y$ be a vertex not on $P$ that is adjacent to both endpoints of $P$ and sees an even number of edges of $P$. Then $G[V(P) \cup \{y\}]$ contains an odd hole as induced subgraph.
\end{lemma}
\begin{proof}
We prove the claim by induction on $n$. Consider the case that $n=4$. If $y$ would be adjacent to exactly one of $v_2$ or $v_3$, then $y$ would see a single edge $\{v_1,v_2\}$ or $\{v_3,v_4\}$ respectively. If $y$ would be adjacent to both $v_2$ and $v_3$, then $y$ would see all three edges of $P$. Since $y$ sees an even number of edges of $P$, it follows that $y$ is only adjacent to $v_1$ and $v_4$. Then $G[V(P) \cup \{y\}]$ induces an odd hole.
\par
In the remaining case we assume that the claim holds for $n' < n$, where $n \geq 6$ and both $n'$ and $n$ are even. Suppose that $y$ sees both edges $\{v_1,v_2\}$ and $\{v_{n-1},v_n\}$, then $P' = (v_2,...,v_{n-1})$ is an induced path on an even number of vertices such that $y$ is adjacent to both of its endpoints and $y$ sees an even number edges in $P'$. By the induction hypothesis $G[V(P') \cup \{y\}]$ contains an odd hole, therefore $G[V(P) \cup \{y\}]$ contains an odd hole as well. If $y$ does not see both $\{v_1,v_2\}$ and $\{v_{n-1},v_n\}$, then assume without loss of generality that $y$ does not see the last edge $\{v_{n-1},v_n\}$. Let $v_j$ for $1 \leq j < n-1$ be the largest index before~$n$ for which $y$ is adjacent to $v_j$. If $j$ is odd, then $G[\{v_j,...,v_n,y\}]$ induces an odd hole. Otherwise $P' = (v_1,...,v_j)$ is an induced path on an even number of vertices, $y$ is adjacent to both of its endpoints, and $y$ sees an even number of edges in $P'$; hence the induction hypothesis applies. In all cases we get that $G[V(P) \cup \{y\}]$ contains an odd hole.
\end{proof}

Before we show the proof that $\Pi_P$ is characterized by rank-4 adjacencies with singleton replacements, we give some intuition for the replacement argument. Suppose we want to replace a vertex $v$ of some odd hole, and we have a set $D$ where each vertex in $D$ is adjacent to both neighbors of $v$ in the hole. Furthermore, the 4-incidence vectors of $D$ sum to the vector of $v$. Then there must exist some vertex in $D$ that sees an even number of edges of the induced path between the neighbors of $v$. This together with \cref{lemma:odd(anti)hole_(odd)evenseen} would result in a graph that contains an odd hole. \cref{fig:matrix} adds to this intuition.

\begin{figure}
    \centering
    \includegraphics[scale=0.8]{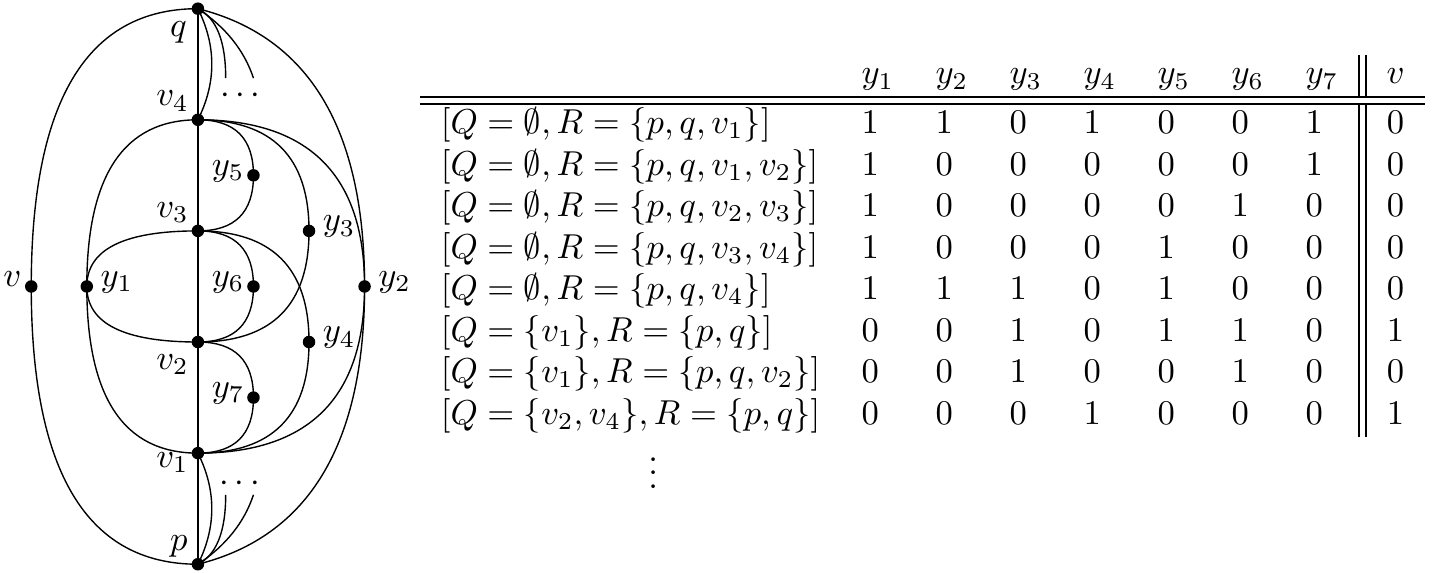}
\caption{Graph~$G$ with vertex cover~$X = \{p,q,v_1,\ldots, v_4\}$, containing an odd hole $H = \{v,p,v_1,...,v_4,q\}$, such that~$\mathrm{inc}^{4,(\emptyset, \{p,q\})}_{(G,X)}(v) = \sum _{i=1}^7 \mathrm{inc}^{4,(\emptyset, \{p,q\})}_{(G,X)}(y_i)$. All edges $\{p,y_i\}$ and $\{q,y_i\}$ for $i \in [7]$ exist, but not all are drawn. The table shows entries of the vectors~$\mathrm{inc}^{4,(\emptyset, \{p,q\})}_{(G,X)}(u \notin X)$. Vertex $y_2$ sees an even number of edges ($\{p,v_1\}$ and $\{v_4,q\}$), and $(v_1,...,v_4,y_2)$ is an odd hole.}
    \label{fig:matrix}
\end{figure}

Let $\Pi_{OH}$ be the set of graphs that contain an odd hole. We show that $\Pi_{OH}$ is characterized by rank-4 adjacencies with singleton replacements.

\begin{theorem}\label{thm:oddholerankc}
$\Pi_{OH}$ is characterized by rank-$4$ adjacencies with singleton replacements.
\end{theorem}
\begin{proof}
Consider some graph $H$ with vertex cover $X$ and let $D \subseteq V(H) \setminus X$ such that $H-D \in \Pi_{OH}$. Let $v$ be an arbitrary vertex in $V(H) \setminus (D \cup X)$ such that $\incs{4}{H}{X}(v) = \sum_{u \in D} \incs{4}{H}{X}(u)$. We show that $H-v-(D\setminus \{v'\}) \in \Pi_{OH}$ for some $v' \in D$.

Let $C$ be an odd hole in $H-D$. If $v \notin V(C)$, then for every $v' \in D$ we have $H-v-(D\setminus \{v'\}) \in \Pi_{OH}$. So suppose that $v \in C$. Let $C = (v,p,v_1,...,v_{n-3},q)$, where $|V(C)|=n$. Consider the induced path $P = (p,v_1,...,v_{n-3},q)$. We have that $v$ is adjacent to $p$ and $q$. Let $D' \subseteq D$ be a set that shares adjacencies with $v$ to $
\{p,q\}$ such that $|D'| \geq 1$ is odd and $\incf{4}{H}{X}{\emptyset}{\{p,q\}}(v) = \sum_{u \in D'} \incf{4}{H}{X}{\emptyset}{\{p,q\}}(u)$. Such set exists by \cref{prop:adjacencyshare}.
Since $C$ is an odd hole, $|V(P)|$ is even. Hence by Lemma~\ref{lemma:odd(anti)hole_(odd)evenseen}, $G[V(P) \cup \{u\}]$ contains an odd hole if there exists some $u \in D'$ that sees an even number of edges of $P$.
Suppose for the sake of contradiction that every vertex in $D'$ sees an odd number of edges of $P$. Let $E_u$ be the set of edges in $P$ that are seen by $u \in D'$. Then $\sum_{u \in D'} |E_u|$ is odd as it is a sum of an odd number of odd numbers. Let $D'_{\{u,w\}} \subseteq D'$ be the set of vertices that see edge $\{u,w\} \in E(P)$. In order to satisfy $\sum_{u \in D'}\incp{\emptyset}{\{p,q\}}(u)[\emptyset,\{p,q,u,w\}] = \incp{\emptyset}{\{p,q\}}(v)[\emptyset,\{p,q,u,w\}] = 0$ over $\mathbb{F}_2$, for $\{u,w\} \in E(P)$, we require $|D'_{\{u,w\}}|$ to be even.
But then $\sum_{e \in E(P)} |D'_e| = \sum_{u \in D'} |E_u|$ would also need to be an even number. This contradicts the fact that $\sum_{u \in D'} |E_u|$ is odd. Therefore there must exist some $u \in D'$ that sees an even number of edges in $P$.
\end{proof}

Let $\Pi_{OAH}$ be the set of graphs that contain an odd anti-hole. Then $\Pi_P = \Pi_{OH} \cup \Pi_{OAH}$. From applications of \cref{lemma:complementgraphprop} and \cref{prop:graphpropertymerge} we get the following.

\begin{corollary}\label{cor:perfect}
Graph properties $\Pi_{OH}$, $\Pi_{OAH}$, and $\Pi_P$ are characterized by rank-4 adjacencies with singleton replacements.
\end{corollary}

\begin{theorem}
\textsc{Perfect Deletion} parameterized by the size of a vertex cover admits a polynomial kernel on $\mathcal{O}(|X|^5)$ vertices.
\end{theorem}
\begin{proof}
By \cref{cor:perfect} we have that $\Pi_P$ is characterized by rank-$4$ adjacencies with singleton replacements. Each graph in $\Pi_P$ contains at least one edge. For each odd hole or odd anti-hole $H$, we have $|V(H)| \leq 2\cdot \vc(H)$. Therefore by \cref{thm:framework} it follows that \textsc{$\Pi_P$-free Deletion} and hence \textsc{Perfect Deletion} parameterized by vertex cover admits a polynomial kernel on $\mathcal{O}(|X|^5)$ vertices.
\end{proof}

A variation of \cref{thm:oddholerankc} presented in \cref{app:evenhole} shows that the set $\Pi_{EH}$ of graphs containing an even hole are characterized by rank-3 adjacencies, which leads to a kernel for \textsc{Even-hole-free Deletion} parameterized by the size of a vertex cover of $\mathcal{O}(|X|^4)$ vertices.

\subsection{AT-free Deletion}\label{sec:atfree}
In his dissertation, K\"ohler~\cite{KOHLER_ATFREE} gives a forbidden subgraph characterization of graphs without asteroidal triples. This forbidden subgraph characterization consists of 15 small graphs on 6 or 7 vertices each, chordless cycles of length at least 6, and three infinite families often called asteroidal witnesses. Let $\Pi_{AT}$ be the set of graphs that contain an asteroidal triple. A technical case analysis leads to the following results.

\begin{theorem}[$\bigstar$]\label{thm:atfreerankc}
$\Pi_{AT}$ is characterized by rank-$8$ adjacencies with singleton replacements.
\end{theorem}

\begin{theorem}\label{thm:atfreekernel}
\textsc{AT-free Deletion} parameterized by the size of a vertex cover admits a polynomial kernel on $\mathcal{O}(|X|^9)$ vertices.
\end{theorem}
\begin{proof}
Every graph in $\Pi_{AT}$ contains at least one edge. By \cref{thm:atfreerankc} it follows that $\Pi_{AT}$ is characterized by rank-8 adjacencies with singleton replacements. Each small graph has at most 7 vertices. For each cycle $C$, we have $|V(C)| \leq 2\cdot \vc(C)$. Finally each asteroidal witness consists of an induced path with 2 or 3 additional vertices, hence there exists $c \in \mathbb{N}$ such that for $G \in \Pi_{AT}$, $|V(G)| \leq c\cdot \vc(G)$. Hence by \cref{thm:framework}, \textsc{AT-free Deletion} parameterized by the size of a vertex cover admits a polynomial kernel on $\mathcal{O}(|X|^9)$ vertices.
\end{proof}

\subsection{Interval Deletion}\label{sec:ivd}
\textsc{Interval Deletion} does not fit in the framework of Fomin et al.~\cite{FOMIN2014468}, since one can show that its forbidden subgraph characterization is not characterized by a finite number of adjacencies. It was shown to admit a polynomial kernel by Agrawal et al.~\cite{AGRAWAL2019}. We show that our framework captures this result. Consider the graph property $\Pi_{IV} = \Pi_{AT} \cup \Pi_{C_{\geq 4}}$, where $\Pi_{AT}$ is the set of graphs that contain an asteroidal triple as in \cref{sec:atfree} and $\Pi_{C_{\geq 4}}$ is the set of graphs that contain an induced cycle of length at least 4. Making a graph $\Pi_{IV}$-free makes it chordal and AT-free, therefore \textsc{$\Pi_{IV}$-free Deletion} corresponds to \textsc{Interval Deletion}.

\begin{theorem}\label{thm:intervalkernel}
\textsc{Interval Deletion} parameterized by the size of a vertex cover admits a polynomial kernel on $\mathcal{O}(|X|^9)$ vertices.
\end{theorem}
\begin{proof}
Every graph in $\Pi_{IV}$ contains at least one edge. By \cref{thm:atfreerankc}, $\Pi_{AT}$ is characterized by rank-8 adjacencies. Furthermore, $\Pi_{C_{\geq 4}}$ is characterized by 3 adjacencies as shown by Fomin et al.~\cite[Proposition 3]{FOMIN2014468}, and therefore by \cref{lemma:c_to_rankc} also by rank-3 adjacencies. Therefore by \cref{prop:graphpropertymerge}, it follows that $\Pi_{IV}$ is also characterized by rank-8 adjacencies.  Each vertex minimal graph in $\Pi_{C_{\geq 4}}$ is a cycle $C$, for which we have $|V(C)| \leq 2\cdot \vc(C)$. Recall that $\Pi_{AT} =  \Pi_S \cup \Pi_{C_{\geq 6}} \cup \Pi_{AW}$. Each vertex minimal graph in $\Pi_S$ contains at most 7 vertices. Finally each asteroidal witness consists of an induced path with 2 or 3 additional vertices, hence there exists $c \in \mathbb{N}$ such that for $G \in \Pi_{IV}$, $|V(G)| \leq c\cdot \vc(G)$. Therefore by \cref{thm:framework}, \textsc{Interval Deletion} parameterized by the size of a vertex cover admits a polynomial kernel on $\mathcal{O}(|X|^9)$ vertices.
\end{proof}

\subsection{(Almost) Wheel-free Deletion}\label{sec:wheelfree}
Let $\Pi_{W_{\geq 3}}$ be the set of graphs that contain a wheel of size at least 3 as induced subgraph. Then \textsc{Wheel-free Deletion} corresponds to \textsc{$\Pi_{W_{\geq 3}}$-free Deletion}. Using a similar argument as in \cref{thm:evenholefreerankc}, we obtain a characterization by rank-4 adjacencies.

\begin{theorem}[$\bigstar$]\label{thm:wheelfreerankc}
$\Pi_{W_{\geq 3}}$ is characterized by rank-4 adjacencies. 
\end{theorem}

Every graph that contains a wheel contains at least one edge. For every wheel $W_n$, we have $|V(W_n)| \leq 2\cdot \vc(W_n)$. Therefore by \cref{thm:framework} we obtain:

\begin{theorem}
\textsc{Wheel-free Deletion} parameterized by the size of a vertex cover admits a polynomial kernel on $\mathcal{O}(|X|^5)$ vertices.
\end{theorem}

It turns out that this good algorithmic behavior is very fragile. Let~$\Pi_{W_{\neq 4}}$ be the set of graphs that contain a wheel of size~$3$, or at least~$5$. Then~\textsc{$\Pi_{W_{\neq 4}}$-free Deletion} corresponds to \textsc{Almost Wheel-free Deletion}. While~$\Pi_{W_{\geq 3}}$ can be characterized by rank-4 adjacencies, the following shows that $\Pi_{W_{\neq 4}}$ is not characterized by adjacencies of any finite rank, and therefore does not fall within the scope of our kernelization framework. 

\begin{theorem}[$\bigstar$]\label{thm:awfnorankc}
$\Pi_{W_{\neq 4}}$ is not characterized by rank-$c$ adjacencies for any $c \in \mathbb{N}$.
\end{theorem}

This is not a deficiency of our framework; we prove that the problem does not have any polynomial compression, and therefore no polynomial kernel, unless \containment. 

\begin{theorem}[$\bigstar$]\label{thm:almostwheelfreenopolykernel}
\textsc{Almost Wheel-free Deletion} parameterized by vertex cover does not admit a polynomial compression unless $\mathrm{coNP} \subseteq \mathrm{NP} / \mathrm{poly}$.
\end{theorem}

This suggests that the condition of being characterized by low-rank adjacencies is the right way to capture kernelization complexity.

\section{Conclusion}\label{sec:conclusion}
We have presented a framework that can be used to obtain polynomial kernels for the \textsc{$\Pi$-free Deletion} problem parameterized by the size of a vertex cover, based on the novel concept of characterizations by low-rank adjacencies. Our framework significantly extends the scope of the earlier framework of Fomin et al.~\cite{FOMIN2014468}. In addition to the examples given in \cref{table:results}, the framework can be applied to obtain kernels for a wide range of vertex-deletion problems. Using the fact that graph properties characterized by low-rank adjacencies are closed under taking a union (\cref{prop:graphpropertymerge}), together with the characterizations by low-rank adjacencies developed here, and characterizations by few adjacencies by Fomin et al.~\cite[Table 1]{FOMIN2014468}, we obtain the following.

\begin{corollary}
Let~$\mathcal{F}$ be a hereditary graph class defined by an arbitrary combination of the following properties: 
    being wheel-free,
    being odd-hole-free,
    being odd-anti-hole-free,
    being even-hole-free,
    being AT-free,
    being bipartite, 
    being $C_{\geq c}$-free for some fixed~$c \in \mathbb{N}$,
    being $H$-minor-free for some fixed graph~$H$,
    being $H$-free for some fixed graph~$H$ containing at least one edge, and
    having a Hamiltonian cycle (respectively, path).
Then the problem of testing whether an input graph~$G$ can be turned into a member of~$\mathcal{F}$ by removing at most~$k$ vertices, has a polynomial kernel parameterized by vertex cover.
\end{corollary}

It would be interesting to see whether the exponents given by \cref{table:results} are tight~(cf.~\cite{GiannopoulouForbiddenMinors}).




\bibliography{references-utf8}


\clearpage

\appendix
\section{Appendix}


\subsection{Even-hole-free Deletion} \label{app:evenhole}
An even hole is a cycle $C_n$, where $n \geq 4$ is an even number. Let $\Pi_{EH}$ be the set of graphs that contain an even hole as induced subgraph. Then the \textsc{Even-hole-free Deletion} problem corresponds to the \textsc{$\Pi_{EH}$-free Deletion} problem.

\begin{lemma}\label{lemma:evenhole_evenseen}
Let $G$ be a graph, $P = (v_1,...,v_n)$ where $n \geq 3$ is odd be an induced path in $G$, and let $y$ be a vertex not on $P$ that is adjacent to both endpoints of $P$ and sees an even number of vertices of $P$. Then $G[V(P) \cup \{y\}]$ contains an even hole as induced subgraph.
\end{lemma}
\begin{proof}
We prove the claim by induction on $n$. Consider the case that $n=3$. Since $y$ is adjacent to both $v_1$ and $v_3$, the only way for it to see an even number of vertices of $P$ is for $y$ not to be adjacent to $v_2$. Therefore $G[V(P) \cup \{y\}]$ induces an even hole and the statement holds.

In the remainder, assume that the claim holds for $n' < n$, where $n \geq 5$ and both $n'$ and $n$ are odd. If $y$ is not adjacent to $v_i$ for any $1 < i < n$, then $G[V(P) \cup \{y\}]$ induces an even hole and the statement holds. Otherwise, since $y$ sees an even number of path vertices it sees at least two more. Let $j > 1$ be the smallest index such that $y$ is adjacent to $v_j$ and let $j' < n$ be the largest index such that $y$ is adjacent to $v_{j'}$. If $j$ is odd, then $G[\{v_1,...,v_j,y\}]$ induces an even hole and the statement holds. If $j'$ is odd, then $G[\{v_{j'},...,v_n,y\}]$ induces an odd hole and the statement holds. Otherwise, both $j$ and $j'$ are even. But then $P' = (v_j,...,v_{j'})$ is a path of odd length at least 3, where $y$ is adjacent to both endpoints and sees an even number of vertices. So by the induction hypothesis, the statement holds. 
\end{proof}

\begin{theorem}\label{thm:evenholefreerankc}
$\Pi_{EH}$ is characterized by rank-3 adjacencies with singleton replacements. 
\end{theorem}
\begin{proof}
Consider some graph $H$ with vertex cover $X$ and let $D \subseteq V(H) \setminus X$ such that $H-D \in \Pi_{EH}$. Let $v$ be an arbitrary vertex in $V(H) \setminus (D \cup X)$ such that $\incs{3}{H}{X}(v) = \sum_{u \in D} \incs{3}{H}{X}(u)$. We show that $H-v-(D\setminus \{v'\}) \in \Pi_{EH}$ for some $v' \in D$.

Let $C$ be an even hole in $H-D$. If $v \notin V(C)$, then for every $v' \in D$ we have $H-v-(D\setminus \{v'\}) \in \Pi_{EH}$. So suppose that $v \in C$. Let $C = (v,p,v_1,...,v_{n-3},q)$, where $|V(C)|=n$. Consider the induced path $P = (p,v_1,...,v_{n-3},q)$. We have that $v$ is adjacent to $p$ and $q$. Let $D' \subseteq D$ be a set that shares adjacencies with $v$ to $
\{p,q\}$ such that $|D'| \geq 1$ is odd and $\incf{3}{H}{X}{\emptyset}{\{p,q\}}(v) = \sum_{u \in D'} \incf{3}{H}{X}{\emptyset}{\{p,q\}}(u)$. Such set exists by \cref{prop:adjacencyshare}.
Since $C$ is an even hole, $|V(P)|$ is odd. Hence by Lemma~\ref{lemma:evenhole_evenseen}, $G[V(P) \cup \{u\}]$ contains an odd hole if there exists some $u \in D'$ that sees an even number of vertices of $P$.
Suppose for the sake of contradiction that every vertex in $D'$ sees an odd number of vertices of $V(P)$. Then also every vertex in $D'$ sees an odd number of vertices of $V(P) \setminus \{p,q\}$. Let $V_u$ be the set of vertices in $V(P)\setminus \{p,q\}$ that are seen by $u \in D'$. Then $\sum_{u \in D'} |V_u|$ is odd as it is a sum of an odd number of odd numbers. 
Let $D'_u \subseteq D'$ be the set of vertices that see vertex $u \in V(P) \setminus \{p,q\}$. 
In order to satisfy $\sum_{u \in D'}
\incp{\emptyset}{\{p,q\}}(u)[\emptyset,\{p,q,u\}] = \incp{\emptyset}{\{p,q\}}(v)[\emptyset,\{p,q,u\}] = 0$ over
$\mathbb{F}_2$
, for $u \in V(P) \setminus \{p,q\}$, 
we require $|D'_u|$ to be even. But then $\sum_{u \in V(P) \setminus \{p,q\}} |D'_u| = \sum_{u \in D'} |V_u|$ would also need to be an even number. This contradicts the fact that $\sum_{u \in D'} |V_u|$ is odd. Therefore there must exist some $u \in D'$ that sees an even number of vertices of $P$.
\end{proof}

Since every graph that contains an even hole has at least one edge and for every hole $H$, we have $|V(H)| \leq 2\cdot \vc(H)$.  By \cref{thm:framework} we obtain:

\begin{theorem}
\textsc{Even-hole-free Deletion} parameterized by the size of a vertex cover admits a polynomial kernel on $\mathcal{O}(|X|^4)$ vertices.
\end{theorem}

\subsection{Omitted material of Section \ref{sec:atfree}}
Let $\Pi_S$ be the set of graphs that contain at least one of these small graphs as induced subgraph, let $\Pi_{C_{\geq 6}}$ be the set of graphs that contain a chordless cycle of length at least 6, and let $\Pi_{AW}$ be the set of graphs that contain a so called asteroidal witness. Then $\Pi_{AT} = \Pi_S \cup \Pi_{C_{\geq 6}} \cup \Pi_{AW}$.  

Fomin et al.~\cite[Proposition 3]{FOMIN2014468} show that $\Pi_{C_{\geq 6}}$ is characterized by 5 adjacencies. By \cref{prop:infiniteset_finitegraphs} we have that $\Pi_S$ is characterized by 6 adjacencies. Therefore we focus our attention on the asteroidal witnesses, which are shown in \cref{fig:aws}.
\begin{definition}\label{def:aws}
The different types of \emph{asteroidal witnesses} (AWs) are constructed as follows.
\begin{itemize}
    \item $\dagger_z$-AW: A graph $G$ such that $V(G) = \{t_l,t_r,t,c\} \cup \{b_1,...,b_z\}$, where $t_l = b_0$ and $t_r = b_{z+1}$, $E(G) = \{\{t,c\}\} \cup \{\{c,b_i\} \mid i \in [z]\} \cup \{\{b_{i-1},b_i\} \mid i \in [z+1]\}$, and $z \geq 2$.
    \item $\ddagger_z$-AW: A graph $G$ such that $V(G) = \{t_l,t_r,t,c_1,c_2\} \cup \{b_1,...,b_z\}$, where $t_l = b_0$ and $t_r = b_{z+1}$, $E(G) = \{\{t,c_1\},\{t,c_2\},\{c_1,c_2\}\} \cup \{\{c_j,b_i\} \mid i \in [z], j\in \{1,2\}\} \cup  \{\{b_{i-1},b_i\} \mid i \in [z+1]\}$, and $z \geq 1$.
    \item $\diamond_z$-AW: A graph $G$ such that $V(G) = \{t_l,t_r,t,c_1,c_2\} \cup \{b_1,...,b_z\}$, where $t_l = b_0$ and $t_r = b_{z+1}$, $E(G) = \{\{t,c_1\},\{t,c_2\}\} \cup \{\{c_j,b_i\} \mid i \in [z], j\in \{1,2\}\} \cup  \{\{b_{i-1},b_i\} \mid i \in [z+1]\}$, and $z \geq 1$.
\end{itemize}
Here $t_l$, $t_r$, and $t$ are called \emph{terminal} vertices, with \emph{top} vertex $t$. These three vertices form an asteroidal triple. Vertices $c$, $c_1$, and $c_2$ are called \emph{center} vertices. Finally $b_i$ for $i \in [z]$ are called \emph{path} vertices.
\end{definition}

\begin{figure}
    \centering
    \includegraphics[scale=0.8]{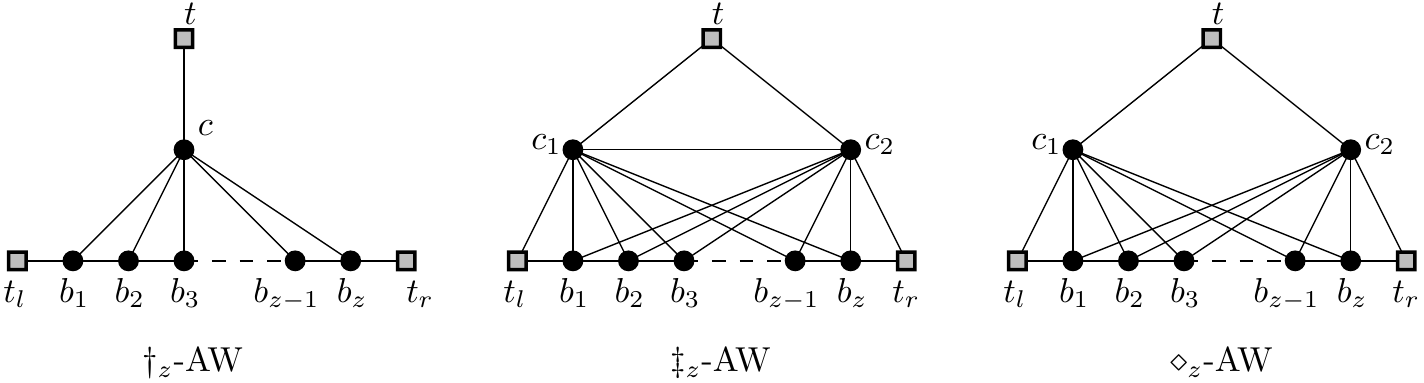}
    \caption{Asteroidal witnesses.}
    \label{fig:aws}
\end{figure}

The difficulty with these asteroidal witnesses lies with the top vertex $t$. If $t$ would be adjacent to a single path vertex $b_i$ for some $i \in [z]$, then the resulting graph no longer contains an asteroidal triple. It is possible to show that, because of this difficulty, no finite adjacency characterization exists for asteroidal witnesses. We show that if we have a set of vertices whose vectors sum to the vector of $t$, then we can find a valid replacement in this set.

\begin{theorem}\label{thm:AWtopreplace}
Let $G$ be a graph with vertex cover $X$, containing $O = x_z$-AW for some $x \in \{\dagger,\ddagger,\diamond\}$ according to Definition~\ref{def:aws} with $t \notin X$. Let $D \subseteq V(G)$ disjoint from $O$ and $X$ such that $\incs{8}{G}{X}(t) = \sum_{y \in D} \incs{8}{G}{X}(y)$. Then $G[(V(O) \cup \{y\}) \setminus \{t\}]$ contains an asteroidal triple for some $y \in D$. 
\end{theorem}

\begin{proof}
Let $G' := G[O \cup D]$ and $X' := X \cap V(G')$. By \cref{note:incsubgraph} we have that $\incs{8}{G'}{X'}(t) = \sum_{y \in D} \incs{8}{G'}{X'}(y)$. Let $C$ be the center vertices of $O$; then $t$ is adjacent to the center vertices in $C$ and not to the other terminals $t_l$ and $t_r$. Let $D' \subseteq D$ be a set of vertices that share adjacencies with $t$ to $\{t_l,t_r\} \cup C$ such that $\incf{8}{G'}{X'}{T}{C}(t) = \sum_{u \in D'} \incf{8}{G'}{X'}{T}{C}(u)$ where $T = \{t_l,t_r\}\cap X'$. Such $D'$ exists by \cref{prop:adjacencyshare}. If there are distinct $w,w' \in D'$ whose vectors are identical,
then $\sum_{u \in D'} \incf{8}{G'}{X'}{T}{C}(u) = \sum_{u \in D' \setminus \{w,w'\}} \incf{8}{G'}{X'}{T}{C}(u)$. Therefore we can remove all duplicate vectors, and hence assume that the vectors for vertices in $D'$ are unique. If there exists $y \in D'$ that is not adjacent to any $b_i$ for $i \in [z]$, then $G[(V(O) \cup \{y\}) \setminus \{t\}]$ induces an $x_z$-AW and we are done.
Otherwise, let $q_y$ be the smallest index such that $y$ is adjacent to $b_{q_y}$ and let $q_y'$ be the largest index such that $y$ is adjacent to $b_{q_y'}$. 
Let $y^* \in D'$ be a vertex that maximizes $q'_{y^*} - q_{y^*}$. If $q'_{y^*} = q_{y^*}$, then $y^*$ is adjacent to a single vertex $b_{q_{y^*}}$. However, since $t$ is not adjacent to this vertex, there must exist $y \in D'$ that is adjacent to the same vertex $b_{q_{y^*}}$ in order to satisfy the vector sum. But then $y$ and $y^*$ have the same projected incidence vector, a contradiction. It remains to deal with the case that~$q'_{y^*} > q_{y^*}$.

If there exists $y \in D'$ (possibly $y^*$ itself) adjacent to $b_{q_{y^*}}$ and $b_{q_{y^*}'}$, but nonadjacent to some $b_i$ with $q_{y^*} + 1 < i < q_{y^*}'-1$, then $t_l$, $t_r$, and $b_i$ form an asteroidal triple in the graph~$G[(V(O) \cup \{y\}) \setminus \{t\}]$, since~$y$ provides a path between~$t_l$ and~$t_r$ avoiding~$N(b_i)$. Therefore suppose there is no such $y$. Then any vertex $y \in D'$ adjacent to $b_{q_{y^*}}$ and $b_{q_{y^*}'}$ is also adjacent to all vertices in between, except possibly $b_{q_{y^*}+1}$ and $b_{q_{y^*}'-1}$. Let $Q = (\{t_l,t_r,b_{q_{y^*}+1},b_{q_{y^*}'-1}\} \setminus N_G(y^*)) \cap X$ and $R = C \cup \{b_{q_{y^*}}, b_{q_{y^*}'}\} \cup (\{b_{q_{y^*}+1},b_{q_{y^*}'-1}\} \cap N_G(y^*))$. We have $|Q|+|R| \leq 8$. The entry $\incf{8}{G'}{X'}{T}{C}(y^*)[Q,R] = 1$ since we defined~$Q$ and~$R$ based on the neighborhood of~$y^*$, but $\incf{8}{G'}{X'}{T}{C}(t)[Q,R] = 0$ since~$t$ is not adjacent to the path vertices~$b_{q_{y^*}}, b_{q'_{y^*}}$.  So there must exist another $y'\in D'$ with $\incf{8}{G'}{X'}{T}{C}(y')[Q,R] = 1$. By the maximality of the range seen by $y^*$, we have $q_{y^*} = q_{y'}$ and $q_{y^*}' = q_{y'}'$. But then $y^*$ and $y'$ have the same projected incidence vector, since they agree on the adjacency to~$b_{q_{y^*}}, b_{q_{y^*}+1}, b_{q'_{y^*}-1}, b_{q'_{y^*}}$ and are adjacent to all vertices in between. This contradicts the assumption that all vectors for vertices in~$D'$ are unique.
\end{proof}

\subsubsection{Proof of Theorem \ref{thm:atfreerankc}}
\begin{theorem*}
The graph property $\Pi_{AT}$ is characterized by rank-$8$ adjacencies with singleton replacements.
\end{theorem*}
\begin{proof}
Consider some graph $H$ with vertex cover $X$ and let $D \subseteq V(H) \setminus X$ such that $H-D \in \Pi_{AT}$. Let $v$ be an arbitrary vertex in $V(H) \setminus (D \cup X)$ such that $\incs{8}{H}{X}(v) = \sum_{u \in D} \incs{8}{H}{X}(u)$. We show that $H-v-(D\setminus \{v'\}) \in \Pi_{AT}$ for some $v' \in D$. The statement for the cases where $H-D \in \Pi_S$ and $H-D \in C_{\geq 6}$ follow from the fact that they are characterized by 6 and 5 adjacencies respectively. In the remainder we have $H-D \in \Pi_{AW}$.

Suppose that $H-D$ contains $O = \dagger_z$-AW for some $z \geq 2$. We do a case distinction on $v$. If $v \notin V(O)$, then for any $v' \in D$ we have $H-v-(D\setminus \{v'\}) \in \Pi_{AT}$, so suppose that $v \in V(O)$.
If $v$ is top terminal $t$, then by \cref{thm:AWtopreplace} we have $H-v-(D\setminus \{v'\}) \in \Pi_{AT}$ for some $v' \in D$.
If $v$ is a center vertex $c$, then $v$ is adjacent to $t$, $b_1$, and $b_2$ and not to $t_l$ and $t_r$. Let $D' \subseteq D$ be the set of vertices share adjacencies with $v$ to $\{t,b_1,b_2,t_l,t_r\}$. Such set $D'$ exists by \cref{prop:adjacencyshare}. Consider any $v' \in D'$, then $\{t,t_l,t_r\}$ still is an asteroidal triple in the graph $H-v-(D\setminus \{v'\})$ and the statement holds. This adjacency argument works for all cases besides the top vertex $t$. We give the set of adjacencies that needs to shared.
For $v = t_l$, the adjacencies to share are $\{b_z,c,t,t_r,b_1\}$. The case $v=t_r$ is symmetric. For $v = b_i$ for some $1 < i < z$, share adjacencies to $\{t,t_l,t_r,b_{i-1},b_{i+1},c\}$. For $v=b_1$, share adjacencies to $\{t,t_r,t_l,b_2,c\}$. Again the case for $v=b_z$ is symmetric.

Finally suppose that $H-D$ contains a $O = \ddagger_z$-AW or $O = \diamond_z$-AW for some $z \geq 1$. If $v \notin V(O)$, then for any $v' \in D$ we have $H-v-(D\setminus \{v'\}) \in \Pi_{AT}$, so suppose that $v \in V(O)$. 
If $v$ is top terminal $t$, then by \cref{thm:AWtopreplace} we have $H-v-(D\setminus \{v'\}) \in \Pi_{AT}$ for some $v' \in D$. For the remaining cases of $v$, again the adjacency argument is sufficient to prove the statement. If $v$ is a center vertex $c_1$, then share adjacencies to $\{t_r,t,t_l\}$. The case $v=c_2$ is symmetric.  For $v=t_l$, share adjacencies to $\{c_2,t_r,t,c_1,b_1\}$. Again the case $v=t_r$ is symmetric. For $v = b_i$ for some $1 < i < z$, share adjacencies to $\{t,b_{i-1},b_{i+1}\}$. For $v=b_1$, share adjacencies to $\{t,t_l,b_{i+1}\}$. The case $v=b_z$ is again symmetric. Note that for all of these case, the set of adjacencies $S$ that some vertex in $D$ needs to share with $v$ contains at most 8 elements. This completes the proof.
\end{proof}

\subsection{Omitted material of Section \ref{sec:wheelfree}}
\subsubsection{Proof of Theorem \ref{thm:wheelfreerankc}}
\begin{theorem*}
The graph property $\Pi_{W_{\geq 3}}$ is characterized by rank-4 adjacencies. 
\end{theorem*}
\begin{proof}
Consider some graph $H$ with vertex cover $X$ and set $D \subseteq V(H) \setminus X$ such that $H-D \in \Pi_{W_{\geq 3}}$. Let $v$ be an arbitrary vertex in $V(H) \setminus (D \cup X)$ such that $\incs{4}{H}{X}(v) = \sum_{u \in D} \incs{4}{H}{X}(u)$. We show that $H-v-(D\setminus D') \in \Pi$ for some $D' \subseteq D$.

Since $H-D \in \Pi_{W_{\geq 3}}$, it contains a wheel $W_n$ for some $n \geq 3$ as induced subgraph. If $v \notin V(W_n)$, then for any $D' \subseteq D$, $H-v-(D\setminus D') \in \Pi_{W_{\geq 3}}$. So suppose $v \in V(W_n)$, then $v$ sees three vertices of $W_n$ that induce a complete graph if $n=3$, or a $P_3$ otherwise. Let $p$, $q$, and $r$ be these vertices.  
Let $D' \subseteq D$ be the set of vertices that share adjacencies with $v$ to $\{p,q,r\}$ such that $\incp{\emptyset}{\{p,q,r\}}(v) = \sum_{u \in D'} \incp{\emptyset}{\{p,q,r\}}(u)$. Such set $D'$ exists by \cref{prop:adjacencyshare}. If $n=3$, then any $d \in D'$ that is adjacent to $p$, $q$, and $r$ is a valid replacement. Suppose that $n > 3$. If $|D'| = 1$, then this vertex must have the same adjacencies as $v$ with respect to the wheel, and hence is a valid replacement. Finally if $|D'| > 1$, then any two vertices in $D'$ together with $p$, $q$, and $r$ induce a $W_4$, since $p$, $q$, and $r$ induce a $P_3$.
\end{proof}

\subsection{Almost Wheel-free Deletion}\label{app:almostwheelfree}
In this section we consider the graph class characterized by the set of forbidden induced subgraphs $\{W_n\mid n=3 \vee n \geq 5\}$. These are the graphs that are almost wheel-free, only $W_4$ might still be in the graph. We give a lower bound for \textsc{Almost Wheel-free Deletion} using a polynomial parameter transformation from \textsc{CNF-SAT} parameterized by the number of variables to the \textsc{Almost Wheel-free Deletion} problem. To this end we give some definitions first.

\begin{definition}
A \emph{polynomial compression} of a parameterized problem $Q \subseteq \Sigma^* \times \mathbb{N}$ into a language $R \subseteq \Sigma^*$ is an algorithm that takes as input an instance $(x,k) \in \Sigma^* \times \mathbb{N}$ and produces in time polynomial in $|x|+k$ a string $y$ such that $|y| \leq p(k)$ for some polynomial $p$ and $y \in R$ if and only if $(x,k) \in Q$.
\end{definition}

Ruling out the existence of a polynomial compression also rules out the existence of a polynomial kernel. We use this fact to construct lower bound proofs. For these proofs we need one more ingredient, namely polynomial parameter transformations.

\begin{definition}
Let $P,Q \subseteq \Sigma^* \times \mathbb{N}$ be two parameterized problems. An algorithm $\mathcal{A}$ is called a \emph{polynomial parameter transformation} (PPT) from $P$ to $Q$ if, given an instance $(x,k)$ of problem $P$, $\mathcal{A}$ produces in polynomial time an equivalent instance $(x',k')$ of problem $Q$. More formally, $(x,k) \in P$ if and only if $(x',k') \in Q$ such that $k' \leq p(k)$ for some polynomial function $p$.
\end{definition}
Note that from the definition above we only require the size of the parameter in the transformed instance to be bounded in some polynomial of the original parameter. The following theorem combines this definition with polynomial compressions.

\begin{theorem}[Theorem 19.2, \cite{KernelizationBook2018}]\label{theorem:ppt}
Let $P$ and $Q$ be two parameterized problems such that there exists a PPT from $P$ to $Q$. If $Q$ has a polynomial compression, then $P$ also has a polynomial compression.
\end{theorem}

By contraposition of the theorem above, if the starting problem $P$ does not have a polynomial compression, then $Q$ also does not have a polynomial compression. This fact can be used to rule out polynomial kernels for parameterized problems. The starting problem for our lower bounds is \textsc{CNF-SAT} parameterized by the number of variables.

\begin{theorem}[Theorem 18.3, \cite{KernelizationBook2018}]\label{thm:cnfsat_parn}
\textsc{CNF-SAT} parameterized by the number of variables $n$ does not admit a polynomial compression unless $\mathrm{coNP} \subseteq \mathrm{NP} / \mathrm{poly}$.
\end{theorem}

With the necessary definitions in place, we prove \cref{thm:almostwheelfreenopolykernel}.

\begin{theorem*}
\textsc{Almost Wheel-free Deletion} parameterized by vertex cover does not admit a polynomial compression unless $\mathrm{coNP} \subseteq \mathrm{NP} / \mathrm{poly}$.
\end{theorem*}
\begin{proof}
We give a polynomial parameter transformation (PPT) from the \textsc{CNF-SAT} problem parameterized by the number of variables. Consider an instance $\phi$ of the \textsc{CNF-SAT} problem with variables $\{x_1,...,x_n\}$ and clauses $\{C_1,...,C_m\}$. In the case that $n \leq 4$, try all possible truth assignments in constant time. If $\phi$ is satisfiable, return an empty graph with an empty vertex cover, which is obviously almost-wheel free. If $\phi$ is not satisfiable, return $G=W_3$ with vertex cover $X = V(G)$ and budget $k=0$. The vertex cover is bounded by a polynomial of $n$ and no vertex can be deleted to make the graph almost wheel-free. In the remainder we consider $n \geq 5$. We construct an instance $(G,X,k)$ of \textsc{Almost Wheel-free Deletion} with vertex cover $X$ and solution size $k$ as follows.
\begin{enumerate}
    \item Starting with an empty graph $G$, for $i \in [n]$, add vertices $v_{x_i}$ and $v_{\neg x_i}$ to $G$. Connect them with an edge.
    \item For $i \in [n]$ add edges $\{v_{x_i},v_{x_{i+1}}\}$, $\{v_{x_i},v_{\neg x_{i+1}}\}$, $\{v_{\neg x_i},v_{x_{i+1}}\}$ and $\{v_{\neg x_i},v_{\neg x_{i+1}}\}$. Here $i+1 = 1$ if $i = n$.
    \item For $i \in [n]$, $j \in [n+1]$, add vertices $u_i^j$ and $w_i^j$. Add edges $\{u_i^j,w_i^j\}$, $\{u_i^j,v_{x_i}\}$, $\{u_i^j,v_{\neg x_i}\}$, $\{w_i^j,v_{x_i}\}$, and $\{w_i^j,v_{\neg x_i}\}$.
    \item Finally consider clause $C_i$, $i \in [m]$. Add vertex $c_i$ to $G$. For every literal $\ell \in C_i$, connect $c_i$ to $v_{\ell}$. Here we assume no clause contains both $x_q$ and $\neg x_q$ for $q \in [n]$. Such clauses are trivially satisfied and can be removed in polynomial time without changing the problem. For every variable $x_q$, $q \in [n]$, that does not correspond to a literal in $C_i$, connect $c_i$ to both $v_{x_q}$ and $v_{\neg x_q}$.
\end{enumerate}
This concludes the construction of $G$. Figure~\ref{fig:cnf2almostwheelfree} shows a sketch of this construction for the clause $x_1 \vee \neg x_2 \vee x_4$. Assign the budget $k=n$.

\begin{figure}
    \centering
    \includegraphics{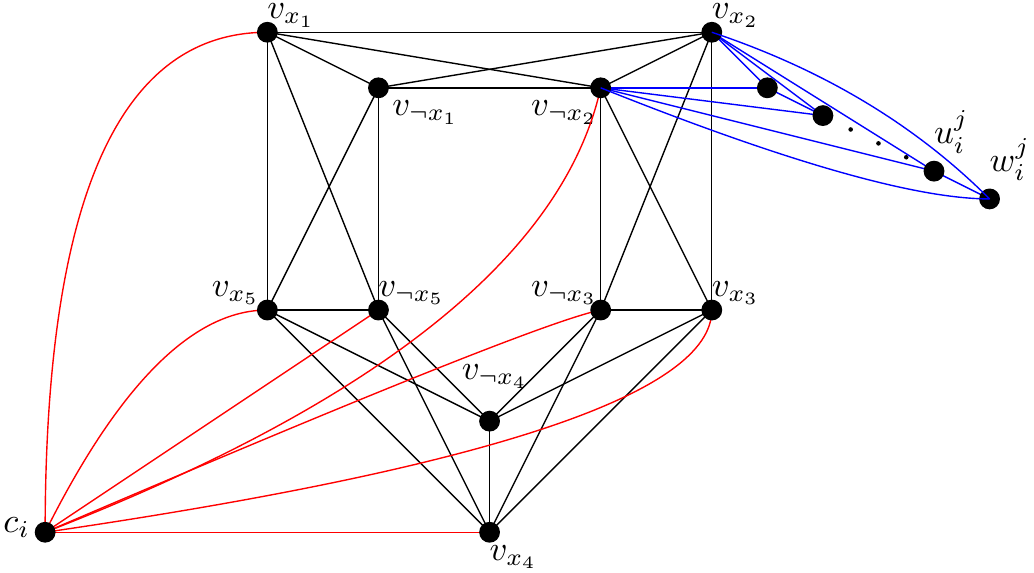}
    \caption{Sketch of transformation for $n=5$ variables. For simplicity the $u$ and $w$ vertices are only (partially) drawn for $x_2$. The clause vertex $c_i$ corresponds to the clause $x_1 \vee \neg x_2 \vee x_4$.}
    \label{fig:cnf2almostwheelfree}
\end{figure}
\begin{note*}
    Vertices added in steps 1-3 form a vertex cover of size $2n + 2n(n+1) = 2n^2 + 4n$.
\end{note*}

Let $X$ be the vertices added in steps 1-3 above. Since $|X| \leq p(n)$ for some polynomial $p(\cdot)$, all that is left to show is the equivalence between the original instance and the vertex-deletion problem $(G,X,k)$. Let $\mathcal{F}$ be the set of almost wheel-free graphs. We show that there is a set $S \subseteq V(G)$ of size at most $k$ such that $G-S$ belongs to $\mathcal{F}$ if and only if the \textsc{CNF-SAT} instance $\phi$ is satisfiable.
\par
($\Rightarrow)$ Suppose $(G,X,k)$ is a \textsc{yes}-instance with solution $S$ of size at most $k$ such that $G-S$ belongs to $\mathcal{F}$.
\begin{claim}
For $i \in [n]$, $S$ contains exactly one of $v_{x_i}$ and $v_{\neg x_i}$.
\end{claim}
\begin{claimproof}
Consider some $i \in [n]$. Suppose $S$ contains neither $v_{x_i}$ nor $v_{\neg x_i}$. For $j \in [n+1]$, the set of vertices $\{v_{x_i},v_{\neg x_i},u_i^j,w_i^j\}$ induce a $W_3$. Since $|S| \leq k$, it follows that $G - S$ would still contain a $W_3$ which contradicts the choice of $S$. It follows that $S$ contains at least one of $v_{x_i}$ and $v_{\neg x_i}$. Since the budget $k=n$, $S$ must contain exactly one of $v_{x_i}$ and $v_{\neg x_i}$.
\end{claimproof}
Consider the following truth assignment $\delta : \{x_1,...,x_n\} \rightarrow \{0,1\}$ such that $\delta(x_i) = 1$ if $v_{x_i} \in S$, and $\delta(x_i) = 0$ if $v_{\neg x_i} \in S$. This is well defined by the claim above.
\begin{claim}
$\delta$ satisfies $\phi$.
\end{claim}
\begin{claimproof}
Suppose $\delta$ does not satisfy $\phi$, then there exists some $i \in [m]$ such that $\delta$ does not satisfy clause $C_i$. Consider the set of vertices $S' = \{v_{\neg \ell_p} \mid v_{\ell_p} \in S, p \in [n], (\ell_p = x_p \vee \ell_p = \neg x_p) \}$. These are the vertices that correspond to the inverse of the truth assignment. The vertices of $S'$ are contained in $G - S$ and induce a cycle of length $n \geq 5$. Since $C_i$ is not satisfied, for each literal $\ell \in C_i$ we have $v_\ell \in S'$. This means vertex $c_i$ is adjacent to vertex $v_\ell$. Now for every variable $x_p$, $p \in [n]$ without a literal in $C_i$, $c_i$ is adjacent to both $v_{x_p}$ and $v_{\neg x_p}$. It follows that $c_i$ is adjacent to every vertex in $S'$, Hence $S' \cup \{c_i\}$ induces a wheel $W_n$. Since $n \geq 5$, this wheel is not of size 4 and hence forbidden, which contradicts the choice of $S$. It follows that $\delta$ satisfies $\phi$.
\end{claimproof}

($\Leftarrow$) In the other direction, assume that there exists a truth assignment $\delta : \{x_1,...,x_n\} \rightarrow \{0,1\}$ that satisfies $\phi$. Let $S = \{v_{\neg x_i} \mid \delta(x_i) = 0, i \in [n] \} \cup \{v_{ x_i} \mid \delta(x_i) = 1, i \in [n] \}$. Obviously $|S| = n$. We show that $G - S$ belongs to $\mathcal{F}$. Let $S' = \{v_{\neg x_i} \mid \delta(x_i) = 1, i \in [n] \} \cup \{v_{ x_i} \mid \delta(x_i) = 0, i \in [n] \}$, which are the vertices that correspond to the inverse of the truth assignment. Again $S'$ forms an induced cycle. $G-S$ is a graph consisting of $S'$ with an independent set of clause vertices $c_i$ for $i\in [m]$ adjacent to vertices in $S'$. Finally each vertex $v_i \in S'$ is part of $n+1$ edge disjoint triangles $\{v_i,u_i^j,w_i^j\}$ for $j \in [n+1]$. Now that we know the structure of $G-S$, we do a case distinction on vertex $v \in V(G-S)$. We show that $v$ cannot be a center vertex of a wheel of size 3 or at least 5.
\begin{itemize}
    \item Case $v \in S'$. Since $S'$ is an induced cycle, $v$ has two neighbors in $S'$, say $x$ and $y$, that are not adjacent. Since the $u$ and $w$ vertices have degree 2 in $G-S$, they cannot be used to create an induced cycle around $v$. Since clause vertices $c_i$, for $i \in [m]$, form an independent set, they cannot appear adjacent in an induced cycle around $v$. Hence the only wheel that could exist around $v$ is $\{v,x,y,c_i,c_j\}$ for some $i,j \in [m]$ such that $c_i$ and $c_j$ are adjacent to $v$, $x$, and $y$, but this is $W_4$ which is not forbidden.
    \item Case $v = c_i$, for some $i \in [m]$. Since $c_i$ only has neighbors in $S'$, and $S'$ is an induced cycle $C_n$, the only way $c_i$ can be a center vertex of an induced wheel is if it is adjacent to every vertex in $S'$. Suppose that $c_i$ is adjacent to every vertex in $S'$, then no literal in clause $C_i$ is satisfied by $\delta$, which contradicts the choice of $\delta$. Therefore $c_i$ cannot be a center vertex of a wheel.
    \item Case $v = u_i^j$ or $v = w_i^j$ for some $i\in [n]$, $j \in [n+1]$. Since $v$ only has two neighbors, it cannot be a center vertex of a wheel.
\end{itemize}
Since $G-S$ does not contain a vertex that could be a center of some wheel $W_p$ for $p = 3$ or $p \geq 5$, it follows that $G-S$ must belong to $\mathcal{F}$. Since we have shown correctness of the PPT, the result follows from Theorem~\ref{theorem:ppt}.
\end{proof}

Let $\Pi_{W_{\neq 4}}$ be the set of graphs that contain a wheel of size 3 or at least 5 as induced subgraph.

\subsubsection{Proof of Theorem \ref{thm:awfnorankc}}

\begin{theorem*}
$\Pi_{W_{\neq 4}}$ is not characterized by rank-$c$ adjacencies for any $c \in \mathbb{N}$. 
\end{theorem*}
\begin{proof}
Suppose for the sake of contradiction that $\Pi_{W_{\neq 4}}$ is characterized by rank-$c$ adjacencies for some $c \in \mathbb{N}$. Construct graph $H$ as follows. Let $i$ be the smallest index such that $q = 2^{i-1}-1 > c$. Vertex cover $X = (v_1,...,v_{n})$ of $H$ is an induced cycle of length $n=2^i-1$. Furthermore $H$ contains a vertex $v \notin X$ that is adjacent to $v_j$ for all $j \in [n]$. Finally $H$ contains an independent set $D$ disjoint from $X \cup \{v\}$. 
The adjacencies of $D$ to $X$ are defined as follows.
Start with an empty set $D$. For $Y \in \binom{[n]}{q}$, add a vertex $d_Y$ to $D$ that is adjacent to $v_j$ if $j \in Y$. 

By Lucas' Theorem~\cite{Fine1947}, we have that $\binom{m}{k}$ is divisible by prime 2, if at least one of the base 2 digits of $k$ is larger than the respective base 2 digit of $m$. Since the binary representation of $n$ and $q$ consists only of ones, we have that $\binom{n-j}{q-j}$ is odd for all $0 \leq j \leq q$. Furthermore, for all $0 \leq j < k \leq q$, we have that $\binom{n-k}{q-j}$ is even.

For disjoint $Q,R \subseteq X$ such that $|Q| + |R| \leq c$, there are $m_{Q,R} = \binom{n-|R|-|Q|}{q-|R|}$ vertices in $D$ adjacent to each vertex in $R$ and to none in $Q$. By our choices of $n$ and $q$, we have that $m_{Q,R}$ is odd if and only if $|Q| = 0$. 

Clearly $H-D$ induces a wheel of size not equal 4 with center~$v$ and cycle $(v_1, \ldots, v_n)$. 

\begin{claim}
For disjoint $Q,R \subseteq X$ such that $|Q| + |R| \leq c$, we have $\incs{c}{H}{X}(v)[Q,R] = \sum_{u \in D} \incs{c}{H}{X}(u)[Q,R]$.
\end{claim}
\begin{claimproof}
We distinguish two cases, depending on whether~$Q$ is empty or not. Note that by construction of~$H$, there are exactly $m_{Q,R}$ vertices of~$D$ which are adjacent to all of~$R$ and none of~$Q$, so there are exactly~$m_{Q,R}$ vertices of~$u \in D$ for which~$\incs{c}{H}{X}(u)[Q,R] = 1$.

If~$Q \neq \emptyset$, then $\incs{c}{H}{X}(v)[Q,R] = 0$ since~$v$ is adjacent to all vertices of~$X$ and therefore to a member of~$Q$. Since~$|Q| > 0$, the value~$m_{Q,R}$ is even as observed above, and therefore zero over $\mathbb{F}_2$. The claim follows.

If~$Q = \emptyset$, then $\incs{c}{H}{X}(v)[Q,R] = 1$ since~$v$ is adjacent to all vertices of~$X$ and therefore to all members of~$R$ and none of~$Q$. As~$|Q| = 0$, the value~$m_{Q,R}$ is odd as observed above, so $\sum_{u \in D} \incs{c}{H}{X}(u)[Q,R] = 1$.
\end{claimproof}

The claim above shows that $\incs{c}{H}{X}(v) = \sum_{u \in D} \incs{c}{H}{X}(u)$. From our assumption that~$\Pi_{W \neq 4}$ is characterized by $c$ adjacencies, it follows that there exists $D' \subseteq D$ such that~$H - v - (D \setminus D') \in \Pi_{W \neq 4}$ and therefore contains a wheel of size unequal to 4. To derive the desired contradiction, we therefore argue that~$H - v$ contains no such wheel. Observe that no vertex~$v \in D$ can be used as the center of a wheel, since~$N_H(v)$ is a proper subset of the induced cycle~$(v_1, \ldots, v_n)$ and therefore acyclic; we use here that~$n > q > c$. Finally for each $j\in [n]$, $v_j$ can only be the center of a wheel of size 4: the graph~$N_H(v_j)$ has a vertex cover~$\{v_{j-1},v_{j+1}\}$ of size two and therefore cannot contain a cycle of length five or more. Furthermore,~$N_H(v_j)$ does not contain a triangle since~$v_{j-1},v_{j+1}$ are non-adjacent and~$D$ is an independent set. Hence there exists no~$D' \subseteq D$ for which~$H - v - (D \setminus D') \in \Pi_{W\neq 4}$, a contradiction.
%
\end{proof}

\end{document}